\documentclass[aip,jmp]{revtex4-1}
\usepackage{amssymb,amsmath,amsthm}
\usepackage{mathrsfs}
\usepackage[pdfpagelabels,plainpages=false]{hyperref}
\usepackage{url}
\usepackage{bbm}   %%% e.g. \mathbbm{1} is a nice notation for the unit matrix
\usepackage{graphicx}
\usepackage{geometry}
\usepackage[all]{pstricks}
\usepackage{pst-func}
\usepackage{pst-plot}
\usepackage{pst-fill}
\usepackage{subfigure}

\usepackage[inline,marginclue,final,silent]{fixme}
\usepackage{paralist}

\usepackage[all]{xy}
\input xy
\xyoption{all}
\xyoption{arc}
\xyoption{line}

\newcommand{\beq}{\begin{displaymath}}
\newcommand{\eeq}{\end{displaymath}}
\newcommand{\beqn}{\begin{equation}}
\newcommand{\eeqn}{\end{equation}}
\newcommand{\beqa}{\begin{eqnarray*}}
\newcommand{\eeqa}{\end{eqnarray*}}
\newcommand{\beqna}{\begin{eqnarray}}
\newcommand{\eeqna}{\end{eqnarray}}

\newcommand{\re}[1]{~(\ref{#1})}
\newcommand{\eq}[1]{~(\ref{#1})}

\newcommand{\N}{\mathbb{N}}

\newcommand{\R}{\mathbb{R}}
\newcommand{\C}{\mathbb{C}}

\newcommand{\ra}{\rightarrow}

\newtheorem{prop}{Proposition}[section]
\newtheorem{thm}[prop]{Theorem}

\newtheorem{lem}[prop]{Lemma}

\theoremstyle{definition} %%% allows plain font for ex, rem and notn

\newtheorem{rem}[prop]{Remark}

\newtheorem{qstn}[prop]{Question}
\begin{document}

\title{Transition probabilities and measurement\\ statistics of postselected ensembles}
\author{
	\firstname{Tobias}
	\surname{Fritz}}
\affiliation{Max Planck Institute for Mathematics}
\email{fritz@mpim-bonn.mpg.de}

\begin{abstract}
It is well-known that a quantum measurement can enhance the transition probability between two quantum states. Such a measurement operates after preparation of the initial state and before postselecting for the final state. Here we analyze this kind of scenario in detail and determine which probability distributions on a finite number of outcomes can occur for an intermediate measurement with postselection, for given values of the following two quantities: (i) the transition probability without measurement, (ii) the transition probability with measurement. This is done for both the cases of projective measurements and of generalized measurements. Among other constraints, this quantifies a trade-off between high randomness in a projective measurement and high measurement-modified transition probability. An intermediate projective measurement can enhance a transition probability such that the failure probability decreases by a factor of up to $2$, but not by more.
\end{abstract}

\maketitle

\section{Introduction}

It is a puzzling property of quantum theory that a measurement on a physical system can change the state of that system in a drastic way. A well-known demonstration of this can be made with polarizers (figure~\ref{polarizers}): upon shining a beam of light onto two orthogonally aligned polarizers, no light at all passes through both of them. However after placing a third polarizer in between the two, such that this new one is not aligned with either of the other two, it is suddenly possible for some light to pass through the whole setup. Hence the middle polarizer, functioning as a projective measurement, has increased the transition probability from zero to a positive value!

Using the simple geometry of a two-state quantum system e.g. in the Bloch sphere, it is not hard to see that the maximal measurement-modified transition probability in this polarizer scenario can be at most $\tfrac{1}{2}$. But what about other cases like $d$-dimensional Hilbert spaces of states---what is the maximal modified transition probability then? Or what if the unmodified transition probability does not vanish? And how do the original and the modified transition probability relate to the statistics of the measurement? These questions are what we are concerned with here---mostly for the case of projective measurements, but also for generalized measurements.

The answers to these questions are statements saying that certain things are possible in quantum theory, while other things are not. Hence these answers might in principle be of interest for further high-precision experimental tests of the quantum formalism. For example, it seems conceivable that models with dynamical wavefunction collapse make different predictions than orthodox quantum theory does.

\begin{figure}
\centering{\includegraphics[width=10cm]{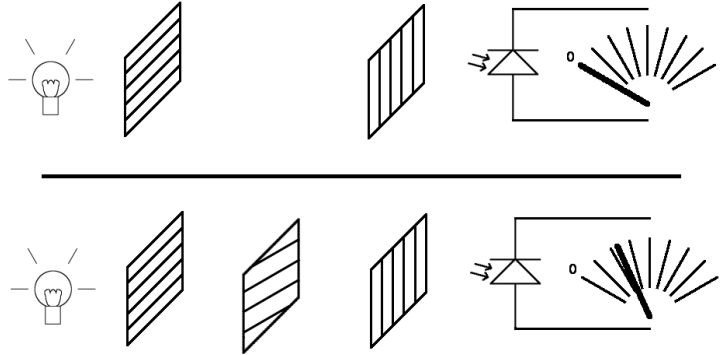}}
\caption{Enhancing transition probabilities by measurement. We regard the first polarizer as state preparation, the last polarizer as postselection, and the intermediate polarizer as a measurement.}
\label{polarizers}
\end{figure}

The natural framework for our considerations is the two-state vector formalism of Aharonov and Vaidman~\cite{Aha}. There, quantum theory becomes time-symmetric by considering two state vectors for a quantum system: an ordinary state vector $|\psi\rangle$ evolving forward in time, and an additional state vector $|\phi\rangle$ evolving backward in time. While $|\psi\rangle$ is to be interpreted as an initial state, $|\phi\rangle$ should be thought of as a target final state: after all other measurements have been done, the experimenter measures the projection operator $|\phi\rangle\langle\phi|$ and records the results of the experiment in case of a positive result, while discarding the whole run in case of a negative result. This procedure is known as \emph{postselection} with respect to $|\phi\rangle$. The polarizer example from above fits into this framework: the first polarizer can be regarded as preparation of the initial state $|\psi\rangle$, while the final polarizer conducts a postselection with respect to a final state $|\phi\rangle$. The measurement statistics obtained by such a procedure are the statistics of the postselected ensemble. Postselected ensembles can show very counterintuitive behavior: examples are the phenomenon that the so-called ``weak value'' of an observable can be bigger than the observable's largest eigenvalue~\cite{AAV}, or the three-boxes thought experiment which displays a high degree of contextuality~\cite{Aha}. Another unexpected property of postselected ensembles has then been found in~\cite{Fri} (see section~\ref{prelim}), and this is the line of investigation to be continued here. The present article should be readable without knowledge of any of the references mentioned.

\paragraph{Synopsis.} Section~\ref{prelim} states the problem studied here and recalls some results from a previous paper~\cite{Fri} about dichotomic measurements with postselection. The main result of the present work then follows in section~\ref{main}, a characterization of all triples $(T,S,P(\cdot))$ allowed in quantum theory, where $T$ is the transition probability without measurement, $S$ is the transition probability with measurement, and $P(\cdot)$ is the statistics of the intermediate $n$-outcome projective (resp. generalized) measurement. After that, section~\ref{discuss} discusses some particular special cases of this result and determines to what extent transition probabilities between quantum states can be enhanced by a projective measurement. The latter two sections frequently refer back to the mathematical appendix~\ref{math}. Finally, section~\ref{conclude} presents a brief conclusion.

\paragraph{Notation.} While Dirac notation appears throughout the main text, but not in the mathematical appendix. Sometimes we use expressions like $\langle\phi|A|\psi\rangle$ also when $A$ is not hermitian. In this case, we stipulate that $A$ acts to the right on the vector $|\psi\rangle$.

\begin{acknowledgments} I would like to thank Andreas Winter for an invitation to visit the Centre for Quantum Technologies, where most of this research has been conducted. I am indebted to Reinhard Werner for stimulating questions during a talk. Furthermore, this work would not have been possible without the excellent research conditions within the IMPRS graduate program at the Max Planck Institute and the invaluable advice provided by my supervisor Matilde Marcolli. Finally, an anonymous referee has kindly provided many highly relevant comments on an earlier version of this manuscript and spotted a gap in the previous proof of proposition~\ref{AStrunc}.
\end{acknowledgments}

\section{Statement of the problem and previous results}
\label{prelim}

\paragraph{Outcome probabilities for ensembles with postselection.} As has also been outlined in the introduction, we consider a quantum system subject to the following kind of procedure:
\begin{enumerate}
\item prepation of some initial state $|\psi\rangle$, 
\item application of a projective (or generalized) measurement with $n$ outcomes,
\item postselection\footnote{For more background on quantum mechanics with postselection and the counterintuitive properties of postselected ensembles, we again refer to~\cite{Aha},~\cite{AAV}.} with respect to some final state $|\phi\rangle$.
\end{enumerate}

We assume that these three consecutive steps happen almost instantaneously, so that the dynamics of the system can be neglected. This is not an essential restriction since we can always take $|\psi\rangle$ to be the actual initial state modified by time evolution until the time of measurement, and similar for $|\phi\rangle$. Also it is no loss of generality to take both $|\psi\rangle$ and $|\phi\rangle$ as pure states, since a mixed state can always be purified by adding an ancilla to the system with which it is entangled (see e.g.~\cite[2.5]{NC}; for the purification of both $|\psi\rangle$ and $|\phi\rangle$, we might have to add two ancillas).

Concerning the intermediate measurement, we will consider the cases of projective measurement and of generalized measurement separately. 

We now calculate the outcome probabilties of the intermediate measurement on the postselected ensemble. The measurement is taken to be defined in terms of Kraus operators $V_k$, $k\in\{1,\ldots,n\}$, with $\sum_k V_k^\dagger V_k=\mathbbm{1}$. It will be assumed for simplicity that the measurement is fine-grained, i.e. that to each outcome $k$ corresponds exactly one Kraus operator $V_k$; this is enough for our main result~\ref{mainthm}, and it should be clear how to extend the following considerations to the general case. With these assumptions, the probability of getting the outcome $k$ in conjunction with successful postselection on the post-measurement state $\frac{V_k|\psi\rangle}{\sqrt{\langle\psi|V_k^\dagger V_k|\psi\rangle}}$ is given by the product of the two respective probabilities as
\beq
\langle\psi|V_k^\dagger V_k|\psi\rangle\cdot\frac{|\langle\phi|V_k|\psi\rangle|^2}{\langle\psi|V_k^\dagger V_k|\psi\rangle}=|\langle\phi|V_k|\psi\rangle|^2.
\eeq
So the cancellation between the normalization of the post-measurement state and the outcome probability gives a surprisingly simple formula for the probability of getting the outcome $k$ in the postselected ensemble:
\beqn
\label{probsgen}
P(k)=\frac{|\langle\phi|V_k|\psi\rangle|^2}{\sum_j|\langle\phi|V_j|\psi\rangle|^2}.
\eeqn
Here, the normalization factor
\beq
\label{success}
S\equiv\sum_j|\langle\phi|V_j|\psi\rangle|^2
\eeq
is the \emph{success probability} of the postselection, i.e. the probability that the final measurement of $|\phi\rangle\langle\phi|$ will give a positive result. We may also regard $S$ as the measurement-modified transition probability. However in order to have a clearer terminology, we will reserve the term ``transition probability'' for $T=|\langle\phi|\psi\rangle|^2$, and refer to $S$ as the ``success probability''.

Note that the formalism is time-reversal invariant in the sense that the roles of the initial state and final state can be interchanged without changing the outcome probabilities or the success probability.

For the case of projective measurements, the Kraus operators should be taken to be a complete set of projection operators,
\beq
V_k=\Pi_k\quad\textrm{ with }\quad\Pi_k^\dagger=\Pi_k,\qquad\Pi_k^2=\Pi_k,\qquad\sum_k\Pi_k=\mathbbm{1}.
\eeq
With this replacement we obtain for the outcome probabilities the Aharonov-Bergmann-Lebowiz formula (eq. (9) in~\cite{Aha}, see also~\cite{ABL}),
\beqn
\label{probs}
P(k)=\frac{|\langle\phi|\Pi_k|\psi\rangle|^2}{\sum_j|\langle\phi|\Pi_j|\psi\rangle|^2},
\eeqn
with the success probability
\beq
S=\sum_j|\langle\phi|\Pi_j|\psi\rangle|^2
\eeq
as normalization factor.

\paragraph{Introducing the problem.} The problem to be solved is the following:

\begin{qstn}
\label{mainqstn}
Given the transition probability $T=|\langle\phi|\psi\rangle|^2\in[0,1]$, which probability distributions $P(\cdot)$ on $\{1,\ldots,n\}$ are outcome distributions of a projective (resp. generalized) measurement for which values of the success probability $S\in(0,1]$?
\end{qstn}

We will only consider the case that the success probability $S$ is strictly positive; for vanishing $S$, the postselected ensemble is empty, and hence the probability distribution $P(\cdot)$ is not defined.

The operational significance of question~\ref{mainqstn} is as follows. All quantities $T$, $S$ and $P(\cdot)$ are in principle experimentally measurable. We imagine that some experiment has provided us with concrete values for these quantities. Then the task is to find a quantum-mechanical model reproducing these particular values, without specifying the Hilbert space dimension in advance, and assuming that the measurement is projective (resp. generalized). Our main result~\ref{mainthm} then tells us directly whether this is possible or not. Now as already mentioned in the introduction, this could be useful for actual high-precision experimental tests of the quantum formalism, and help to distinguish e.g. models of dynamical wavefunction collapse~\cite{Pearle} from orthodox quantum theory, where wavefunction collapse happens instantaneously. 

\paragraph{Previous results.} The surprising results of~\cite{Fri} have been a strong motivation for the present work. There, the case $T=0$ and $n=2$ has been treated in section 2, and it was found that the only possibility is given by $P(1)=P(2)=\tfrac{1}{2}$, independently of $S$. This is actually easiest to see on the level of amplitudes, where it follows from
\beq
0=\langle\phi|\psi\rangle=\langle\phi|\Pi_1|\psi\rangle+\langle\phi|\Pi_2|\psi\rangle,
\eeq
so that the two probabilities for measuring $1$ or $2$ are given by, respectively,
\beq
P(1)=\frac{|\langle\phi|\Pi_1|\psi\rangle|^2}{|\langle\phi|\Pi_1|\psi\rangle|^2+|\langle\phi|\Pi_2|\psi\rangle|^2}=\frac{1}{2},\qquad P(2)=\frac{|\langle\phi|\Pi_2|\psi\rangle|^2}{|\langle\phi|\Pi_1|\psi\rangle|^2+|\langle\phi|\Pi_2|\psi\rangle|^2}=\frac{1}{2}.
\eeq
Intuitively, this means that a dichotomic projective measurement with postselection which is orthogonal to the inital state is guaranteed to be a perfectly unbiased random number generator.

\section{Main results}
\label{main}

Using the elementary mathematical results listed in appendix~\ref{math}, we are now ready to answer question~\ref{mainqstn} in generality.

\begin{thm}
\label{mainthm}
\begin{enumerate}
\item A given probability distribution $P(\cdot)$ with given $T\in[0,1]$ and $S\in(0,1]$ can occur via a 
projective measurement if and only if all the inequalities
\beqn
\label{Pineqs}
\sqrt{P(k)}\leq\sqrt{\frac{T}{S}}+\sum_{j\neq k}\sqrt{P(j)}\quad\forall k,\qquad \sqrt{\frac{T}{S}}\leq\sum_k\sqrt{P(k)}\leq\frac{1}{\sqrt{S}}
\eeqn
hold.
\item With a generalized measurement, any combination of values for $P(\cdot)$, $T$ and $S$ can occur.
\end{enumerate}
\end{thm}

\begin{proof} We start with the proof in the projective measurement case. The main idea here is to use the completeness relation $\sum_k \Pi_k=\mathbbm{1}$ in order to obtain an identity for amplitudes
\beq
\langle\phi|\psi\rangle=\sum_k\langle\phi|\Pi_k|\psi\rangle
\eeq
and then translate this into conditions on the probabilities\re{probs}. To this end, we can apply lemma~\ref{proph} to
\beq
z_k=\langle\phi|\Pi_k|\psi\rangle,\:\: k=1,\ldots,n,\qquad z_{n+1}=-\langle\phi|\psi\rangle.
\eeq
For then upon setting $x_k\equiv\sqrt{P(k)S}=|\langle\phi|\Pi_k|\psi\rangle|$ for $k=1,\ldots,n$, and defining $x_{n+1}=\sqrt{T}$, it follows that the left-most inequalities of\re{Pineqs} are necessary, as well as the first inequality of the second formula. 

The remaining inequality follows from two applications of the Cauchy-Schwarz-inequality as follows:
\beq
\sum_k|z_k|=\sum_k |\langle\phi|\Pi_k|\psi\rangle|\leq\sum_k\sqrt{\langle\phi|\Pi_k|\phi\rangle}\cdot\sqrt{\langle\psi|\Pi_k|\psi\rangle}\leq\sqrt{\sum_k\langle\phi|\Pi_k|\phi\rangle}\cdot\sqrt{\sum_k\langle\psi|\Pi_k|\psi\rangle}=1,
\eeq
as was to be shown.

To see that the inequalities\re{Pineqs} taken together are also sufficient for the existence of a quantum-mechanical model, we again set $x_k$ to be given by the square roots of the unnormalized probabilities as $x_k\equiv\sqrt{P(k)S}$ for $k=1,\ldots,n$, and again define $x_{n+1}=\sqrt{T}$. Then once more by~\ref{proph}, some compatible $z_k$'s with $\sum_{k=1}^{n+1}z_k=0$ can now assumed to be given, and they also satisfy $\sum_{k=1}^n|z_k|=\sum_{k=1}^nx_n\leq 1$ by the assumption\re{Pineqs}. Now one can use lemma~\ref{ampls} to obtain the states on $\C^n$ which are given by
\beq
|\psi\rangle=\sum_{k=1}^n\psi_k|k\rangle,\qquad |\phi\rangle=\sum_{k=1}^n\phi_k|k\rangle
\eeq
in conjunction with the projection operators $\Pi_k=|k\rangle\langle k|$ for $k=1,\ldots,n$. Then $\sqrt{P(k)S}=|\langle\phi|\Pi_k|\psi\rangle|$ and $T=|\langle\phi|\psi\rangle|^2$ both hold by construction. The requirement $S=\sum_k|\langle\phi|\Pi_k|\psi\rangle|^2$ is automatic by normalization of the probability distribution $P(\cdot)$. This ends the proof in the projective measurement case.

In the generalized measurement case, we will construct $|\psi\rangle$, $|\phi\rangle$ and $V_k$ which reproduce the given data. We first choose any unit vectors $|\psi\rangle$ and $|\phi\rangle$ satisfying $|\langle\phi|\psi\rangle|^2=T$. Now in Hilbert space of dimension at least $n$, it is possible to find a complete set of mutually orthogonal projectors $\Pi_k$ such that $P(k)=\langle\psi|\Pi_k|\psi\rangle$; if one would measure these on the ensemble defined by the initial state $|\psi\rangle$ without postselection, one would obtain the given distribution $P(\cdot)$. Now fix some unit vector $|\phi'\rangle$ with $|\langle\phi'|\phi\rangle|^2=S$. Then for those $k$ with $P(k)>0$, there exists a unitary $U_k$ which maps the unit vector $P(k)^{-1/2}\Pi_k|\psi\rangle$ to $|\phi'\rangle$. We take $V_k\equiv U_k\Pi_k$, while setting $V_k\equiv 0$ for those $k$ with $P(k)=0$. What we have thus constructed is a generalized measurement in which the post-measurement state is always $|\phi'\rangle$; this guarantees that the measurement statistics on the initial state $|\psi\rangle$ and the postselection are probabilistically independent. Hence by construction, the desired statistics $P(\cdot)$, $T$ and $S$ have been reproduced.
\end{proof}

It is possible to rewrite the inequalities\re{Pineqs} in a more convenient form. Since the left-most inequality holds for all $k$ if and only if it holds for that $k$ for which $P(k)$ is largest, it is enough to require
\beq
2\sqrt{\max_k P(k)}\leq\sqrt{\frac{T}{S}}+\sum_k\sqrt{P(k)}
\eeq
In terms of the diversity indices\eq{divindex},\eq{maxindex}
\beqn
\label{moms}
D_\infty\equiv\frac{1}{\max_k P(k)},\qquad D_{1/2}\equiv\left(\sum_k\sqrt{P(k)}\right)^2
\eeqn
we can see that the inequalities\re{Pineqs} are in fact equivalent to
\beqn
\label{chain}
\boxed{\frac{2}{\sqrt{D_\infty}}-\sqrt{D_{1/2}}\leq\sqrt{\frac{T}{S}}\leq\sqrt{D_{1/2}}\leq\frac{1}{\sqrt{S}}}
\eeqn
so that the dependence on the distribution $P(\cdot)$ is only through the dependence on the quantities $D_\infty$ and $D_{1/2}$. By\eq{holderD}, the allowed interval for $\sqrt{T/S}$ is always non-empty.

\begin{rem}
\begin{enumerate}
\item The proof of the theorem shows that it is sufficient to employ Hilbert spaces of dimension at most $n$. For projective measurements, this is clearly best possible. For generalized measurements however, the number of outcomes is not related to the Hilbert space dimension, and so it might be interesting to study how much the existence of a quantum-mechanical model depends on Hilbert space dimension. Since the proof above still involves many arbitrary choices, it seems conceivable that one can cover a sizeable part of the space of triplets $(P(\cdot),T,S)$ e.g. by qubit models.
\item The diversity indices $D_\infty$ and $D_{1/2}$ are simply the exponentials of the min-entropy and the R\'enyi $\tfrac{1}{2}$-entropy, respectively:
\beq
H_\infty=\log D_\infty,\qquad H_{1/2}=\log D_{1/2}.
\eeq 
\item The right-most inequality in\eq{chain} states that
\beq
S\leq \frac{1}{D_{1/2}}.
\eeq
Intuitively, this means that high randomness in the measurement implies a low success probability. So in order to achieve a high success probability, one needs to choose a projective measurement with not too much randomness on the postselected ensemble.

\item A very nice example of how to control transition amplitudes by measurements is the Aharonov-Vardi effect~\cite{AV}, a variant of the quantum Zeno effect. The observation is that any given quantum dynamics $|\psi(t)\rangle$ can be approximately simulated by starting with the initial state $|\psi(t_0)\rangle$ and conducting projective measurements $|\psi(t_n)\rangle\langle\psi(t_n)|$ at the times $t_n\equiv t_0+n\cdot\delta t$, with $n\in\N$. Aharonov and Vardi~\cite{AV} have shown in particular that for $\delta t\ra 0$, the probability of obtaining any target state $|\psi(t_f)\rangle$ at any final time $t_f$ approaches unity. Since such a sequence of projective measurements can also be seen as a single generalized measurement, this illustrates part (b) of the theorem.
\end{enumerate}
\end{rem}

\section{Discussion}
\label{discuss}

Let us now look at some specific cases of theorem~\ref{mainthm}(a). So in this section, ``measurement'' always means ``projective measurement''.

\paragraph{Case $T=0$ with $S$ arbitrary.} This is the case that has been studied in~\cite{Fri} for $n=2$. As long as we allow the success probability $S$ to be arbitrarily small, all that remains are the inequalities
\beqn
\label{triangle}
\sqrt{P(k)}\leq\sum_{j\neq k}\sqrt{P(j)}\quad\forall k
\eeqn
For $n=2$, this reads $\sqrt{P(1)}\leq\sqrt{P(2)}$ and $\sqrt{P(2)}\leq\sqrt{P(1)}$, implying that $P(1)=P(2)=\tfrac{1}{2}$. Hence a dichotomic measurement with postselection which is orthogonal to the initial state is guaranteed to be a perfectly unbiased random number generator (see section~\ref{prelim}). 
The $n=3$ case is illustrated in figure~\ref{fig1}; one obtains a circular disk within the probability simplex. This can be shown from\eq{triangle} by squaring the inequalities, rearranging, and then squaring again while taking care of the signs. This eventually leads to the quadratic inequalities
\beq
\left(P(1)-P(2)-P(3)\right)^2\leq 2P(2)P(3)\quad +\quad\textrm{cyclic permutations}
\eeq
for the circular shape of the quantum region in figure~\ref{fig1}. Also, just as it should due to the result for the $n=2$ case, the $n=3$ region intersects with any side of the triangle in exactly the middle of that side. So whenever the final state is orthogonal to the initial state, any intermediate projective measurement with three outcomes needs to show statistics lying in this disk.

For arbitrary $n\geq 2$, one can at least say that the $\sqrt{P(k)}$ always lie in a certain subset of $\R^n$ which is the convex cone defined by the inequalities\eq{triangle}. Since these $n$ inequalities are linearly independent in $\R^n$, for dimensional reasons this convex cone is a simplex, i.e. the conical hull of $n$ linearly independent extreme rays. One can calculate the $m$th extreme ray by requiring all inequalities except for the $m$th one to be saturated. Solving the ensuing system of linear equations shows that the $m$th extreme ray $y^m$ has the coordinates
\beq
y^m_j=1+(2-n)\delta_{jm}
\eeq
Hence for any $P(\cdot)$ satisfying\eq{triangle}, one can find non-negative real numbers $\lambda_m$ such that
\beq
\sqrt{P(k)}=\sum_m\lambda_my^m_k.
\eeq

\begin{figure}
\psset{unit=130pt}
\centering{
\begin{pspicture}(0,0)(1,1)
\psline(0,0)(.5,.866)
\psline(.5,.866)(1,0)
\psline(0,0)(1,0)
\pscircle[fillstyle=solid,fillcolor=lightgray](.5,.289){0.289}
\rput(-.05,0){$1$}
\rput(1.05,0){$2$}
\rput(.5,.916){$3$}
\end{pspicture}}
\caption{The quantum-mechanical region within the probability simplex for three measurement outcomes, $T=0$ (orthogonal postselection), and arbitrary success probability $S$. This is a ternary plot, i.e. each vertex stands for a definite outcome, and each point inside the triangle represents a probability distribution over the vertices.}
\label{fig1}
\end{figure}
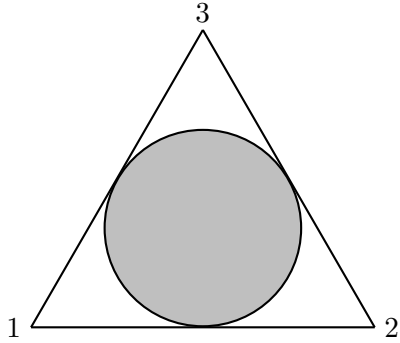

\paragraph{Case $S$ fixed, $P(\cdot)$ fixed.} The first two inequalities of\eq{chain} define an interval of possible values for the transition probability $T$. This can be interpreted as follows: by knowing the behavior of the system with measurement, it is possible to predict something about how the system would behave without measurement.

\paragraph{Case $T>0$ fixed, $P(\cdot)$ fixed, $S$ arbitrary.} Here, it is possible for any $P(\cdot)$ to find some appropriately small success probability $S$ such that all inequalities in\re{chain} hold (e.g. $S=\frac{T}{D_{1/2}}$), so no constraints abound. This is one reason why it is important to always consider $S$ as an additional parameter.

\paragraph{Case $n=2$ with $T$ and $S$ unspecified.} Here, the two probability values $P(1)$ and $P(2)$ determine each other uniquely, so let us write $P(1)=p$ and $P(2)=1-p$. Then the inequalities are
\beqn
\label{binary}
\big|\sqrt{p}-\sqrt{1-p}\big|\leq\sqrt{\frac{T}{S}}\leq \sqrt{p}+\sqrt{1-p}\leq\frac{1}{\sqrt{S}}
\eeqn
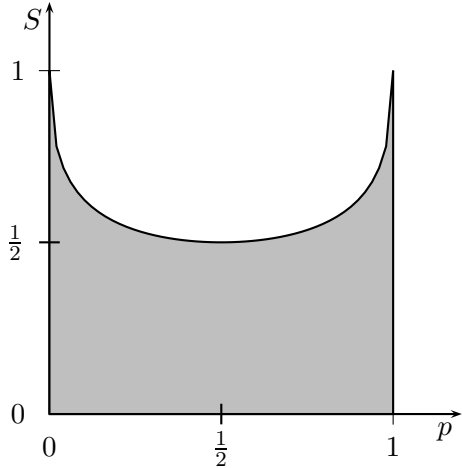
\begin{figure}
\psset{unit=130pt}
\centering{\begin{pspicture}(-.23,-.23)(1.2,1.2)
\rput(-.05,1.15){$S$}
\rput(1.15,-.05){$p$}
\psaxes{->}(0,0)(0,0)(1.2,1.2)
\psset{fillstyle=solid,fillcolor=lightgray}
\pscustom{
\psplot{0}{1}{1 1 2 x 1 x sub mul sqrt mul add div}
\psline(1,1)(1,0)
\psline(1,0)(0,0)
\psline(0,0)(0,1)}
\psline(.5,-.03)(.5,.03)
\rput(.5,-.1){$\tfrac{1}{2}$}
\psline(-.03,.5)(.03,.5)
\rput(-.1,.5){$\tfrac{1}{2}$}
\end{pspicture}}
\caption{For $n=2$ (dichotomic measurement), the possible quantum-mechanical success probabilities $S$ as a function of the outcome probability $p$. High randomness in the measurement decreases the maximal probability of successful postselection, i.e. the maximal measurement-modified transition probability.}
\label{fign2}
\end{figure}
The projection of this into the $p$-$S$-plane, where only the last inequality is relevant, is shown in figure~\ref{fign2}. For fixed $S$, some sections of the quantum region are graphed in figure~\ref{fign2A}. The first two inequalities of\re{binary} define the upper and lower boundary curves in these figures, while the third inequality leads to vertical cuts whenever $S>\tfrac{1}{2}$.

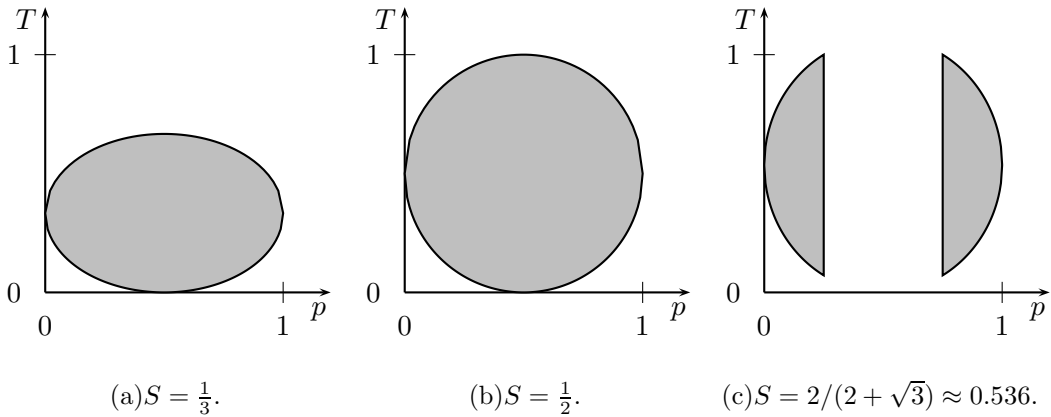
\begin{figure}
\psset{unit=090pt}
\centering{\subfigure[$S=\tfrac{1}{3}$.]{\begin{pspicture}(-.23,-.23)(1.2,1.2)
\rput(-.08,1.15){$T$}
\rput(1.15,-.08){$p$}
\psaxes{->}(0,0)(0,0)(1.2,1.2)
\psset{fillstyle=solid,fillcolor=lightgray}
\pscustom{
\psplot{.5}{1}{1 3 div sqrt x sqrt 1 x sub sqrt sub mul dup mul}
\psplot{1}{0}{1 3 div sqrt 1 x sub sqrt x sqrt add mul dup mul}
\psplot{0}{.5}{1 3 div sqrt 1 x sub sqrt x sqrt sub mul dup mul}}
\end{pspicture}
}
\subfigure[$S=\tfrac{1}{2}$.]{\begin{pspicture}(-.23,-.23)(1.2,1.2)
\rput(-.08,1.15){$T$}
\rput(1.15,-.08){$p$}
\psaxes{->}(0,0)(0,0)(1.2,1.2)
\psset{fillstyle=solid,fillcolor=lightgray}
\pscustom{
\psplot{.5}{1}{1 2 div sqrt x sqrt 1 x sub sqrt sub mul dup mul}
\psplot{1}{0}{1 2 div sqrt 1 x sub sqrt x sqrt add mul dup mul}
\psplot{0}{.5}{1 2 div sqrt 1 x sub sqrt x sqrt sub mul dup mul}}
\end{pspicture}
}
\subfigure[$S=2/(2+\sqrt{3})\approx 0.536$.]{\begin{pspicture}(-.23,-.23)(1.2,1.2)
\rput(-.08,1.15){$T$}
\rput(1.15,-.08){$p$}
\psaxes{->}(0,0)(0,0)(1.2,1.2)
\psset{fillstyle=solid,fillcolor=lightgray}
\pscustom{
\psplot{1}{.75}{.5359 sqrt 1 x sub sqrt x sqrt add mul dup mul}
\psline(.75,1)(.75,.072)
\psplot{.75}{1}{.5359 sqrt x sqrt 1 x sub sqrt sub mul dup mul}}
\pscustom{
\psplot{0}{.25}{.5359 sqrt 1 x sub sqrt x sqrt add mul dup mul}
\psline(.25,1)(.25,.072)
\psplot{.25}{0}{.5359 sqrt 1 x sub sqrt x sqrt sub mul dup mul}}
\end{pspicture}}}
\caption{Again $n=2$ (dichotomic measurement). These plots show the quantum-mechanical region for $(p,T)$ for some values of $S$. The vertical cuts for $S>\tfrac{1}{2}$ appear due to figure~\ref{fign2}. For $S\ra 1$, these cuts rapidly approach the $p=0$ and $p=1$ axes. One possible interpretation is that knowing the system behavior with measurement (i.e. $S$ and $p$) lets us say something about system behavior without measurement (i.e. $T$).} 
\label{fign2A}
\end{figure}

\paragraph{The $T$-$S$-region.} How does the transition probability relate in general to the probability of successful postselection? To study this, it is best to consider the inequalities in the form\re{chain}. Figure~\ref{figAS} shows an illustration of the following proposition.

\begin{prop}
\label{AStrunc}
For a given number of outcomes $n$, some success probability $S$ can appear in quantum theory together with some transition probability $T$ if and only if
\beqn
\label{AS}
\boxed{\frac{T}{n}\leq S\leq\frac{T+1}{2}}
\eeqn
\end{prop}

\begin{proof}
Again it is first shown that these inequalities are necessary. Since $D_{1/2}\leq n$, the second inequality in\re{chain} implies that
\beq
T\leq nS.
\eeq
For proving the second inequality of\eq{AS}, we distinguish two cases. If, firstly, the left-most term of\eq{chain} is non-negative, we can square the left-most inequality of\eq{chain} and use it as follows:
\beq
\frac{T}{S}+\frac{1}{S}\stackrel{(\ref{chain})}{\geq}\frac{4}{D_\infty}+2D_{1/2}-4\sqrt{\frac{D_{1/2}}{D_\infty}}=2\left[\frac{1}{D_\infty}+\left(\sqrt{D_{1/2}}-\frac{1}{\sqrt{D_\infty}}\right)^2\right],
\eeq
so that the desired result follows from lemma~\ref{strange}(a). The second case is that the left-most term of\eq{chain} is negative, which means that $D_{1/2}\geq\frac{4}{D_\infty}$. If $D_\infty\leq 2$, we have therefore $D_{1/2}\geq 2$, so that $S\leq\tfrac{1}{2}$ by $S\leq\frac{1}{D_{1/2}}$. Finally if $D_\infty\geq 2$, then\eq{holderD} also shows that $D_{1/2}\geq 2$, giving the same conclusion $S\leq\tfrac{1}{2}$. In all cases, the second inequality of\eq{AS} has therefore been verified.

For checking sufficiency of\re{AS}, consider first the case that $\frac{T}{n}\leq S\leq T$. Then by lemma~\ref{strange}(b), it follows that the first inequality of\re{chain} holds automatically. The possible values for $D_{1/2}$ are given by the closed interval $[1,n]$. Hence it is possible to find some value for $D_{1/2}$ in this interval which also satisfies\re{chain} whenever $\frac{1}{\sqrt{S}}\geq 1$, which holds trivially, and $\sqrt{\frac{T}{S}}\leq\sqrt{n}$, which is true by assumption. This ends the proof in this case.

It remains to prove sufficiency when $T\leq S\leq\frac{T+1}{2}$. Here, it is in fact enough to consider probability distributions $P(\cdot)$ supported on two elements, which brings us effectively down to the dichotomic case $n=2$ from equation\re{binary}. By $\sqrt{\frac{T}{S}}\leq 1$, the middle inequality of\eq{binary} is automatic, so one only needs to take care of the remaining two. These in turn can be written as
\beq
1-2\sqrt{p(1-p)}\leq\frac{T}{S},\qquad 1+2\sqrt{p(1-p)}\leq\frac{1}{S}
\eeq
Upon choosing $p$ such that the second inequality is saturated, one finds that the first inequality is satisfied as long as $\frac{T+1}{S}\geq 2$.
\end{proof}

So this result gives clear bounds on how much a measurement can enhance or reduce a transition probability. It has been found that a measurement can reduce a transition probability by a factor which is given by the number of outcomes of the measurement. This becomes intuitive when one thinks of the measurement---with outcomes discarded---as a decoherence process which can drive the system's state towards a totally mixed state or highly mixed state.

The situation for enhancing transition probabilities by measurement is very different. We can rewrite the first inequality of\eq{AS} more conveniently in terms of the \emph{failure probabilities} $1-T$ and $1-S$, where it reads
\beq
\boxed{1-S\geq\frac{1-T}{2}}
\eeq
Hence, a measurement can lower the probability that a desired state transition fails by a factor of up to $2$, but not by more. The proof above has shown that this enhancement can already be achieved by a two-outcome measurement. This is again intuitive in terms of the decoherence due to measurement: for creating a successful transition, it would be useless to try to measure a projection operator with support outside of the linear span $\mathrm{lin}\left\{|\psi\rangle,|\phi\rangle\right\}$. Therefore, a transition-enhancing measurement should have non-vanishing probability on exactly two outcomes.

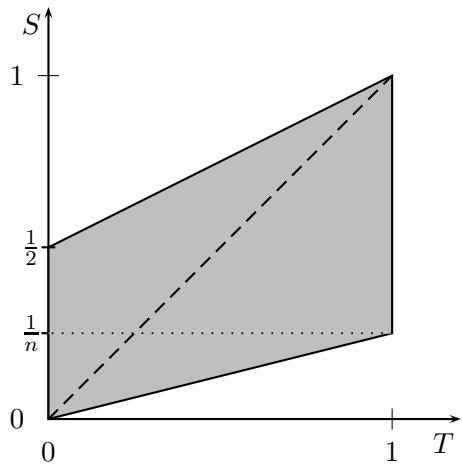
\begin{figure}
\psset{unit=130pt}
\centering{\begin{pspicture}(-.23,-.23)(1.2,1.2)
\rput(-.05,1.15){$S$}
\rput(1.15,-.07){$T$}
\psaxes{->}(0,0)(0,0)(1.2,1.2)
\psset{fillstyle=solid,fillcolor=lightgray}
\pscustom{
\psplot{0}{1}{1 x add 2 div}
\psline(1,1)(1,.25)
\psplot{1}{0}{x 4 div}
\psline(0,0)(0,.5)}
\psline[linestyle=dashed](0,0)(1,1)
\psline[linestyle=dotted](-.02,.25)(1,.25)
\psline(-.02,.25)(0,.25)
\rput(-.05,.25){$\frac{1}{n}$}
\psline(-.02,.5)(.02,.5)
\rput(-.05,.5){$\frac{1}{2}$}
\end{pspicture}}
\caption{The quantum region of transition probabilities: $T$ is the transition probability without measurement, while $S$ is the transition probability with $n$-ary projective measurement. All points above the dashed diagonal $S=T$ represent a measurement-enhanced transition probability.}
\label{figAS}
\end{figure}

\section{Conclusion}
\label{conclude}
It is a well-known phenomenon that measurements influence transition probabilities between quantum states. In this article, we have conducted a systematic study of this phenomenon and determined how it relates to the outcome distribution of the intermediate measurement on the corresponding postselected ensemble. It has been found that a given probability distribution for a projective measurement can appear in conjunction with a given transition probability and a given success probability of the postselection if and only if certain inequalities hold. These inequalities depend on the probability distribution only through its min-entropy and its R\'enyi $\tfrac{1}{2}$-entropy. Furthermore, no conditions at all abound if the measurement is allowed to be any generalized quantum measurement.

As a consequence of these results, it was possible to bound the enhancement of transition probabilities by projective measurements. The maximal enhancement can be achieved with two-outcome measurements and is such that the failure probability decreases by a factor of $2$.

\appendix
\section{Mathematical appendix}
\label{math}

Here we collect various elementary mathematical facts which are referenced from the main text.

It is known (see e.g.~\cite{Pin}) that a finite sequence of non-negative real numbers, $x_1,\ldots,x_n$ is the sequence of edge lengths of a polygon in the Euclidean plane if and only if the inequalities
\beqn
\label{ineq1}
x_k\leq\sum_{j\neq k} x_j
\eeqn
hold. Since they generalize the triangle inequality, these inequalities are known as \emph{polygon inequalities}; on the other hand, the triangle inequality immediately implies that the polygon inequalities are necessary for the existence of a polygon with these edge lengths. See e.g.~\cite{HSS} for another occurence of the polygon inequalities in quantum information theory.

This geometrical statement directly implies the following:

\begin{lem}
\label{proph}
Given non-negative real numbers $x_1,\ldots,x_n$, there exist complex numbers $z_1,\ldots,z_n$ with
\beq
|z_k|=x_k,\qquad \sum_k z_k=0
\eeq
if and only if the inequalities
\beqn
\label{ineqh}
x_k\leq\sum_{j\neq k} x_j
\eeqn
hold.
\end{lem}

\begin{lem}
\label{ampls}
For $n\geq 2$ and any $z\in\C^n$, there exist $\psi,\phi\in\C^n$ with
\beq
||\psi||_2=1=||\phi||_2,\qquad\quad\overline{\psi}_k\phi_k=z_k\quad\forall k=1,\ldots,n
\eeq
if and only if the inequality
\beqn
\label{l1norm}
\sum_k|z_k|\leq 1
\eeqn
holds.
\end{lem}

\begin{proof}
Necessity of\re{l1norm} is nothing but the Cauchy-Schwarz inequality:
\beq
\sum_k|z_k|=\sum_k|\psi_k|\cdot|\phi_k|\leq\sqrt{\sum_k|\psi_k|^2}\cdot\sqrt{\sum_k|\phi_k|^2}\leq 1
\eeq
That\eq{l1norm} is also sufficient for the existence of such $\psi$ and $\phi$ will be shown by induction on $n$. Note that the phases of $z_k$ can be changed arbitrarily without altering the (non-)existence of such vectors, hence we may as well assume that all $z_k$ are non-negative real numbers. We now prove the statement for the initial case $n=2$. By the assumptions $z_1,z_2\geq 0$ and $z_1+z_2\leq 1$, it is implied that $|z_1-z_2|\leq 1$, and therefore it is possible to find angles $\alpha$ and $\beta$ such that
\begin{align*}
\cos(\alpha+\beta)=\cos\alpha\cos\beta-\sin\alpha\sin\beta\stackrel{!}{=}z_1-z_2,\\
\cos(\alpha-\beta)=\cos\alpha\cos\beta+\sin\alpha\sin\beta\stackrel{!}{=}z_1+z_2.
\end{align*}
Hence the two vectors
\beq
\psi=\left(\begin{array}{c}\cos\alpha\\\sin\alpha\end{array}\right),\qquad \phi=\left(\begin{array}{c}\cos\beta\\\sin\beta\end{array}\right)
\eeq
have all the required properties.

The induction step is a simple rescaling argument. Given $z_1,\ldots,z_{n+1}\geq 0$ with $\sum_k z_k\leq 1$, define $z'_1,\ldots,z'_n$ as
\beq
z'_k\equiv\frac{z_k}{1-z_{n+1}},\quad k=1,\ldots,n.
\eeq
(We may assume $z_{n+1}\neq 1$ e.g. by reordering the $z_k$'s.) Then by induction assumption, we can find $\psi',\phi'\in\C^n$ with $||\psi'||_2=||\phi'||_2=1$ and $\overline{\psi}'_k\phi'_k=z'_k$. Now the two vectors
\beq
\psi_k\equiv\left\{\begin{array}{cl}\psi'_k\sqrt{1-z_{n+1}}&\textrm{ for }k=1,\ldots,n\\\sqrt{z_{n+1}}&\textrm{ for }k=n+1\end{array}\right.,\qquad\phi_k\equiv\left\{\begin{array}{cl}\phi'_k\sqrt{1-z_{n+1}}&\textrm{ for }k=1,\ldots,n\\\sqrt{z_{n+1}}&\textrm{ for }k=n+1\end{array}\right., 
\eeq
do indeed have the desired properties $||\psi||_2=||\phi||_2=1$ and $\overline{\psi}_k\phi_k=z_k$, which also finishes the induction step.
\end{proof}

The following fact can also be regarded as a special case of the H\"older inequality, but since a direct proof is extremely simple, we have included it here.

\begin{lem}
\label{holder}
Let $x\in\R_{\geq 0}^n$ with $\sum_k x_k=1$. Then,
\beq
\left(\sum_k\sqrt{x_k}\right)^2\cdot\max_k x_k\geq 1
\eeq
\end{lem}

\begin{proof}
This is easily shown by a direct calculation:
\beq
1=\sum_k \sqrt{x_k}\cdot\sqrt{x_k}\leq\left(\sum_k\sqrt{x_k}\right)\cdot\max_k\sqrt{x_k}
\eeq
so that squaring gives the desired result.
\end{proof}

\paragraph{Diversity indices.} A diversity index~\cite{Div}, as used for example in biostatistics, is a function that assigns to each probability distribution a real number which is intended to measure a sort of effective cardinality contained in the probability distribution. In other words, a diversity index is an exponentiated entropy. Like in the main text, let $P(\cdot)$ be a probability distribution on $\{1,\ldots,n\}$. For each $q\in(0,\infty)$, one obtains a diversity index $D_q$ by defining 
\beqn
\label{divindex}
D_q(P)\equiv\left(\sum_k P(k)^q\right)^{\frac{1}{1-q}}
\eeqn
For $q=1$, this has to be understood as $\lim_{q\ra 1}D_q$, which is the exponentiated Shannon entropy. In a similar way, it is possible to define $D_0(P)$, which is the cardinality of the support of $P$, and $D_\infty(P)$, which turns out to be
\beqn
\label{maxindex}
D_\infty(P)=\frac{1}{\max_k P(k)}.
\eeqn
The relevant quantities for us are going to be $D_{1/2}$ and $D_\infty$. When $P$ is the uniform distribution on $n$ elements, we have $D_q(P)=n$ for all $q$.

In this notation, we get a simple reformulation of lemma~\ref{holder}:
\beqn
\label{holderD}
D_\infty\leq D_{1/2}.
\eeqn

\begin{lem}
\label{strange}
\begin{enumerate}
\item
\beq
\frac{1}{D_\infty}+\left(\sqrt{D_{1/2}}-\frac{1}{\sqrt{D_\infty}}\right)^2\geq 1
\eeq
\item
\beq
\frac{2}{\sqrt{D_\infty}}-\sqrt{D_{1/2}}\leq 1
\eeq
\end{enumerate}
\end{lem}

\begin{proof}
\begin{enumerate}
\item 
Let $k_0\in\{1,\ldots,n\}$ be such $P(k_0)$ is the highest probability in the distribution, i.e. $P(k_0)=\max_k P(k)$. Then,
\beq
\frac{1}{D_\infty}+\left(\sqrt{D_{1/2}}-\frac{1}{\sqrt{D_\infty}}\right)^2=P(k_0)+\left(\sum_{k\neq k_0}\sqrt{P(k)}\right)^2\geq P(n)+\sum_{k\neq k_0}P(k)=1,
\eeq
\item This is trivial by $D_\infty\geq 1$ and $D_{1/2}\geq 1$.
\end{enumerate}
\end{proof}

\end{document}